\documentclass[10pt, conference, compsocconf]{IEEEtran}

\usepackage{lineno}
\usepackage{graphicx}
\usepackage{amsmath}
\usepackage{amsthm}
\usepackage{amssymb}
\usepackage{pst-all}
\usepackage{algorithmic,algorithm}
\usepackage{epstopdf}
\usepackage{subfloat}

\newtheorem{thm}{Theorem}
\newtheorem{prop}[thm]{Proposition}
\newtheorem{defn}[thm]{Definition}

\newtheorem{lem}[thm]{Lemma}

\newtheorem{clm}[thm]{Claim}
\newtheorem{lemma}{Lemma}

\newcommand{\tuple}[1]{\mbox{$\langle #1 \rangle$}}

\renewcommand{\implies}{\Rightarrow}

\newcommand{\mand}{\wedge}

\newcommand{\defines}{\equiv}
\newcommand{\suchthat}{\mbox{$\cdot$}}

\newcommand{\trace}{\mbox{$\cal T$}}
\newcommand{\features}{\mbox{${\cal F}$}}
\newcommand{\components}{\mbox{${\cal C}$}}
\newcommand{\spec}{\overline{\features}}
\newcommand{\arch}{\overline{\components}}
\newcommand{\powerset}{\wp}

\newcommand{\scope}{\mbox{$\overline{\features}$}}
\newcommand{\platform}{\mbox{$\overline{\components}$}}
\newcommand{\spl}{\mbox{$\Psi$}}

\newcounter{cstcount}
\begin{document}
%
\title{Formalizing Traceability and Derivability in Software Product Lines}


\author{\IEEEauthorblockN{Swarup Mohalik, Ramesh S., Jean-vivien Millo}
\IEEEauthorblockA{India Science Lab\\
General Motors, TCI\\
Bangalore, India\\
Email: \{swarup.mohalik,ramesh.s,jean.v\}@gm.com}
\and
\IEEEauthorblockN{Shankara Narayanan Krishna, Ganesh Narwane}
\IEEEauthorblockA{Dept. of Computer Science and Engineering\\
I.I.T., Powai\\
Mumbai, India\\
\{krishnas,ganeshk\}@cse.iitb.ac.in}
}

\maketitle

\begin{abstract}
In the literature, the definition of product in a Software Product Line (SPL) is 
based upon the notion of consistency of the constraints, imposed by variability and traceability relations on the elements of the SPL. In this paper, we contend that consistency does not model the natural semantics of the implementability relation  between problem and solution spaces correctly. Therefore, we define when a feature can be {\em derived} from a set of components . Using this, we define a product of the SPL by a $\langle$specification, architecture$\rangle$ pair, where all the features in the specification are derived from the components in the architecture. This notion of derivability is formulated in  a simple yet expressive, abstract model of a productline with traceability relation. We then define a set of SPL analysis problems  and show that these problems can be encoded as Quantified Boolean Formulas. Then, QSAT solvers like QUBE can be used to solve the analysis problems. We illustrate the methodology on a small fragment of a realistic productline.
\end{abstract}


%
\IEEEpeerreviewmaketitle

\section{Introduction}
\label{intro}

Software Product Line (SPL) is a development framework to jointly design a family of closely related software {\em products} in an efficient and cost-effective manner. Every SPL is built upon a collection of features and components. Each individual product is {\em specified} by a subset of features

Each product in the family is {\em specified} by a set of features drawn from a collection common to the family, and is {\em implemented} by an architecture comprising a set of reusable components selected from a collection of {\em basic assets} which are developed once for the entire family. 

There are two key orthogonal aspects of an SPL, namely, {\em variability} and {\em traceability}. While variability introduces different choices (termed {\em variation points}) within the artifacts in system development, such as specifications, architectures and components, traceability relates the variation points together across the artifacts. Since variability introduces complex constraints among the variation points, managing variability in large industrial SPLs is quite complex and has given rise to a number of analysis problems. These have been the focus of SPL research in the recent years.  A comprehensive survey of these analysis problems and their solutions can be found in Benavides et al.\cite{benavides10-is}. 

On the other hand, we observe that traceability and its implications have not been studied in as much depth in the literature. In the following, we mention the few works addressing traceability as a primary aspect. It is defined in \cite{Beuche2004} as one of the  four important characteristics of a variability model, namely, consistency, visualization, scalability and traceability.  A variability management model that focuses on the traceability aspect between the notion of problem and solution spaces is presented in \cite{Berg2005}. Anquetil et al.\cite{Anquetil2008} formalize the traceability relations across problem and solution space and also across domain and product engineering. In \cite{302409}, the notion of {\em product maps} is defined which is a matrix giving the relation between features and products.  Consistency analysis of product maps is presented in \cite{Eisenbarth2002}. Zhu et al.\cite{Zhu2006} define a traceability relation from requirement to feature and also from feature to architecture with consistency analysis. \cite{10.1109/ICCSA.2007.59} presents a consistency verification method between feature model and architecture model. Metzger et al.\cite{Metzger2007} differentiate SPL variability and product variability and then present a framework based on OVM by Pohl et al.\cite{1095605} to perform checks for consistency, liveness, commonness, realizability (completeness), and flexibility (soundness).

One of the central concepts of the SPL analyses in the above-mentioned works is that of a {\em product}. It is defined through the notion of consistency between a collection of features and components and the constraints imposed by variability and traceability. In this report, we contend that consistency does not model the natural semantics of the implementability relation  between problem and solution spaces correctly. It allows components and features to coexist without any conflict, but it also allows cases where the features may not be derivable from the components. Hence, the SPLs can be shown to allow products where the components are not related to the features in a more intuitive notion of traceability. Therefore, we define when a feature can be {\em derived} from a set of components. Using this, we define a product of the SPL by a $\langle$specification, architecture$\rangle$ pair, where all the features in the specification are derived from the components in the architecture. This definition of products is tighter than the existing "consistency" based definitions. 

Another contribution of the report is a simple yet expressive, abstract model of a productline where we formally define the derivability notion through the traceability relation. We then define a set of SPL analysis problems. Some of these problems are already addressed in earlier works but are redefined in the light of the new concepts. The others are new and arose because of the separation of problem and solution space linked through traceability. We show that these problems, in general, can be encoded as Quantified Boolean Formulas(QBF) and QSAT solvers\cite{qsat-solver} can be used to solve the problems. We illustrate the methodology on a small fragment of a realistic productline.

The summary of our contributions in this report are the following:
\begin{enumerate}
    \item A new definition for SPL products based on a notion of derivability of feature specifications from component architectures. The traceability relation plays the central role in this definition.
	 \item A simple, abstract semantic model of SPL with traceability. The model abstracts out the details from the existing descriptions of SPL in the literature and allows us to define the core concepts in a formal and concise manner.
	\item A set of analysis problems in the SPL, some of which are known but cast anew in the light of the new definitions, and others that are novel.
	\item A solution method for the analysis problems which is based on QBF encoding and QSAT solving. This is necessitated by the nature of some of the analysis problems and is in contrast to the SAT based solving methods generally employed for the extant SPL analyses.

\end{enumerate}
\paragraph{Outline of The Report}
In the following section, we introduce a case study of Entry Control Product Line (ECPL) from the automotive domain. This is used as a running example throughout the rest of the report. The formal model of an SPL with traceability is described in Section \ref{sec_formal}. It introduces the central notion of derivability and the analyses we would like to carry out in SPL. In Section \ref{sec_encoding}, we show how the analysis problems can be encoded in QBF. The results of the analyses using QSAT on the ECPL case study is presented in Section \ref{sec_casestudy1}. Finally, we conclude in Section \ref{sec_con} with a summary of the report and some future directions. The proof of the main result relating the analysis problems and QBF formulae is given in the appendix.

\section{The Entry Control Product Line (ECPL)}
\label{sec_casestudy}

We introduce a fragment of a typical Entry Control Product Line (ECPL) used in the automotive industry. It will be used to illustrate the concepts throughout the report and as a case study in Section~\ref{sec_casestudy1}. The entry control system comprises all the features involved in the controlling of door locking/unlocking in a car. In this study, we focus on the following subset:
\begin{itemize}
\item Manual lock: controls the locking/unlocking through manual lever presses
\item Power lock: controls the locking/unlocking according to key button press, courtesy switch press and sill button press. 
\item Door lock: controls automatic locking of doors when the vehicle starts.
\item Door relock: controls automatic relocking of doors in case of pick up/drop and drive.
\end{itemize}

\paragraph{The ECPL feature diagram}

Figure \ref{fig_fm} presents the feature diagram of the ECPL (a la Czarnecki~\cite{CzarneckiHE05}). The dark gray boxes are features of the ECPL. 
The light gray boxes are parameters modeled as features. The Power lock feature is mandatory. Manual lock is optional. When it is present, the Power lock feature is excluded. The Door lock feature is optional and can be triggered either when gear is shifted out of park or when car speed reaches a predefined value. The Door relock feature is optional. The car should have either a manual or an automatic transmission. Manual transmission disallows the ``park options" of Door lock since there is no park gear in a manual gearbox.

\begin{figure}[hbpt]
 \begin{center}
   {\includegraphics[scale=0.35]{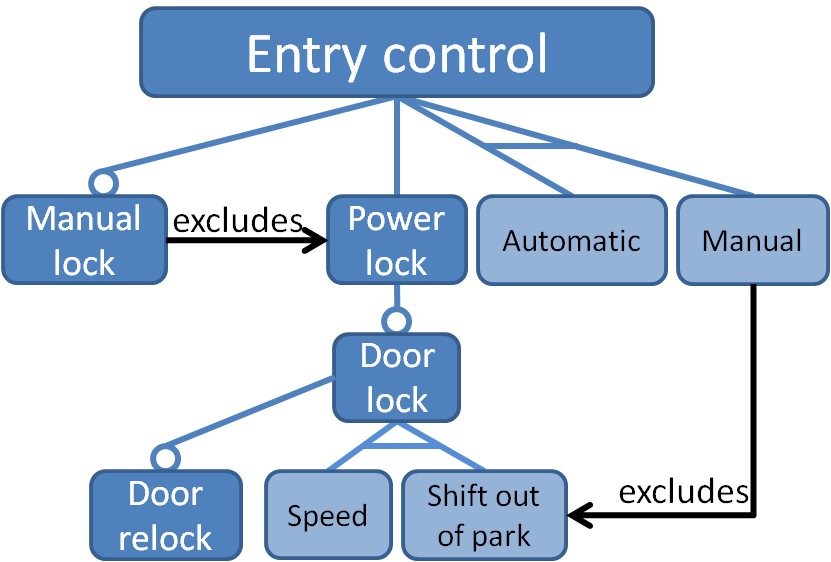}}
  \caption{The feature diagram of the ECPL.} 
  \label{fig_fm}
  \end{center}
\end{figure}

\paragraph{The ECPL architectural diagram}

Figure \ref{fig_mam} represents the platform of ECPL using a notation called Modal Architectural Model (abbreviated as MAM). 
It is a simplified form of EASEL by \cite{Hoek2000} and yet preserves the essential notion of variability central to the product line. The platform is composed of three components: $Door$ $lock$ $manager$, $Power$ $lock$, and $Auto$ $lock$. The first is mandatory but the two others are optional (denoted by dotted boxes). The system has seven ``in" ports (dark squares) and three ``out" ports  (light squares). The interconnections between external and internal ports connect ports of the same type but internal interconnection connect complementary ports (``out" port to ``in" port). The signals ``Transmission in Park" and ``Speed" are alternatives. Similarly ``Automatic" and ``Manual" inform the system on the type of transmission.

\begin{figure}[hbpt]
 \begin{center}
   {\includegraphics[scale=0.4]{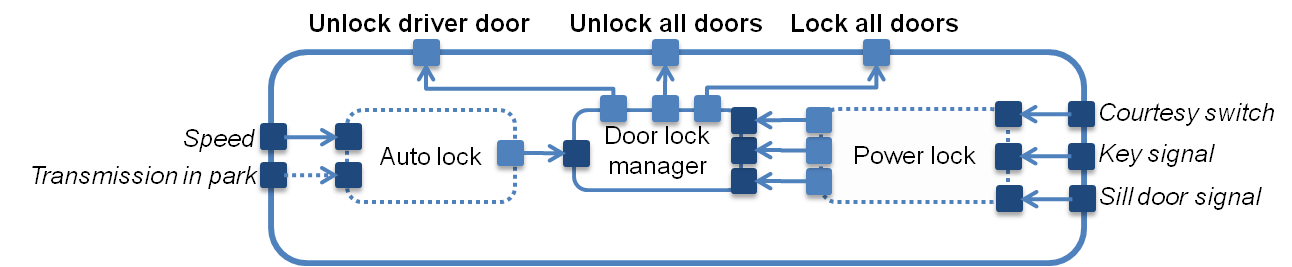}}
  \caption{The platform of the ECPL.} 
  \label{fig_mam}
  \end{center}
\end{figure}

$Auto$ $lock$ component requires two global input signals while $Power$ $lock$ component requires five. They provide lock/unlock command signals to $Door$ $lock$ $manager$. The command provided by $Power$ $lock$ component depends upon manual action, and the command provided by $Auto$ $lock$ component is according to the requirements of the features Door lock and Door relock.

The $Door$ $lock$ $manager$ component arbitrates the lock/unlock command signals from $Auto$ $lock$ and $Power$ $lock$ and forwards them to the global outputs depending upon a calibration (1/Unlock all doors, 2/Unlock Driver door, 3/Lock all doors).

\paragraph{The traceability relations of the ECPL}
To avoid confusion between the homonymous features and components ({\em Automatic}, {\em Manual}, and {\em Speed}), we will, in the sequel, prefix the labels with $f\_$ or $c\_$ respectively. Table \ref{tab_req} presents the required components to implement each feature. 

\begin{table}[htbp]
	\centering
	\begin{tabular}{|c|c|c|}
	\hline
	Feature & Component \\
	\hline
	\hline
	Power lock& Door lock manager\& Power lock\\
	Door lock& Auto lock\\
	Door relock& Auto lock\\
	f\_Automatic& c\_Automatic\\
	f\_Manual& c\_Manual\\
	Shift out of Park& Gear in Park\\
	f\_Speed& c\_Speed\\
	\hline
	\end{tabular}
	\caption{Each feature requires component(s)}
	\label{tab_req}
\end{table}

Table \ref{tab_prov} presents the features provided by the architectural elements.

\begin{table}[htbp]
	\centering
	\begin{tabular}{|c|c|c|}
	\hline
	Component/Interconnection & Feature\\
	\hline
	\hline
	 Door lock manager \& Power lock& Power lock \\
	 c\_Automatic& f\_Automatic\\
	 c\_Manual& f\_Manual\\
	 Auto lock & Door lock \& Door relock \\
	 Gear in park  & Shift out of park \\
	 \hline
	 c\_Speed& f\_Speed\\
	 \hline	
	\end{tabular}
	\caption{The architectural elements provide some features}
	\label{tab_prov}
\end{table}

\section{Model of SPL : Traceability and Implementation}
\label{sec_formal}
In this section, we propose a model of the software productline 
making the traceability relation explicit and define an implementation
relation between architectures and specifications based on traceability.

\subsection{Modeling Decisions}  
In ~\cite{Metzger2007}, the traceability relation is given as a set of arbitrary propositional constraints over the components and features. In the current report, we impose a fairly natural structure on the traceability relation, consisting of a {\em provides} and a {\em requires} function for each feature. 
This is inspired by the points of view of the suppliers and integrators (OEMs). 
Suppliers usually would package one or more features in a component, which is captured by the provides relation. On the other hand, integrators start with a set of features which requires a set of components for implementation.

Importantly, the implementations are related to the specifications only when they can be derived using the traceability relation. Consider a simple SPL consisting of a feature $f$ and a component $c$, but without any traceability relation between $f$ and $c$. According to analyses such as in \cite{Metzger2007}, since $\{f, c\}$ is consistent (in a propositional logic), it is considered as a product. Clearly, it is not natural. On the other hand, if $f$ was provided by $c$, then $\{f, c\}$ would be a natural product.

Another novel point in our model is the notion of approximate implementation (Covers). In the literature, the definition of implementation is usually exact: we need the components that provide exactly the same set of features in a specification. However, since many components are pre-built by the suppliers, there may not be a choice suitable for an exact implementation. For example, if the OEM wants a feature of ABS (Anti-lock Braking) and the supplier has packaged both ABS and TC(Traction Control) in one component, the OEM has to choose this component which covers (but does not exactly implement) the specification of ABS. 

\subsection{Formal Model}
Let $\features$ be a set of features. A subset of $\features$ is called a {\em specification}. The scope of an SPL is a collection of specifications: $\spec \subseteq \powerset(\features)$. The specifications are implemented using a set of (reusable) components $\components$. Each subset of $\components$ is called an {\em architecture}. An SPL platform consists of a set of architectures: $\platform \subseteq \powerset(\components)$.\footnote{The representation of specification and platform is semantic in nature. Syntactic representation of these may use FODA diagrams, MaMs or a variety of notations in the literature. In general, one can have implicit representations through constraints on the features and components; we will adopt this view in the following sections.}

A traceability relation $\trace$ connects the features and components: $\trace$ is specified as a pair $\tuple{prov, req}$ where $prov$ and $req$ are maps $\features \rightarrow \powerset(\powerset(\components))$. Through the traceability relation we capture the {\em sufficient} ($prov(.)$) and {\em necessary} ($req(.)$) conditions to implement a feature. When $prov(f) = \{C_1, C_2\}$, we interpret it as the fact that the set of components $C_1$ (also, $C_2$) provides the implementation of the feature $f$. On the other hand, when $req(f) = \{D_1, D_2\}$, we interpret as the fact that the implementation of the feature $f$ requires  the set of components $D_1$ or the set of components $D_2$. 

\begin{defn}
An SPL $\spl$ is defined as a triple $\tuple{\scope, \platform, \trace}$, where $\scope$ is the scope, $\platform$ is the platform and $\trace$ is the traceability relation.
\end{defn}

In the ECPL case study, $\features$ contains the nine features of Figure \ref{fig_fm} and the ECPL scope $\scope$ contains eight specifications. For illustration, we choose the following specifications: $spec_{1}=\{Power\,lock, f\_Automatic\}$ and $spec_{2}=\{Power\,lock, f\_Automatic, Door\,lock, Shift\,out\,of\,park,$\\ $Door\,relock\}$. The top-most feature {\em Entry control} is in every specification and is not mentioned explicitly.

In ECPL, $\components$ contains the three components of Figure \ref{fig_mam} and the twelve interconnections which are also modeled as components. Note that the mandatory interconnections are in every architecture and are not mentioned explicitly. The ECPL platform $\platform$ contains nine architectures which can be extracted from the ECPL platform. Again, for illustration, we select two architectures $arch_1=\{Door\,lock\,manager\}$ or $arch_2=\{Door\,lock\,manager, Power\,lock,$\\ $c\_Automatic, Auto\,lock, Transmission\,in\,park\}$.

The traceability relation in ECPL is given through the Tables  \ref{tab_req}($req(.)$) and \ref{tab_prov}($prov(.)$). For example, the $Auto\,lock$ component provides the features $Door\,lock$ and $Door\,relock$. Each of these features requires only $Auto\,lock$ component.

\noindent
The main concept of implementability in $\spl$ is defined as follows: a feature is implemented by an architecture (set of components in $\platform$) if the architecture provides the feature and simultaneously fulfills the mandatory requirements of the feature.
\begin{defn}[Implements]
Given an SPL $\spl = \tuple{\scope, \platform,  \trace}$, 
$implements_{\spl}(C, f)$ {\rm if} 
$\exists C_1 \in prov(f), C_2 \in req(f) \suchthat C_2 \subseteq C_1 \subseteq C.$
 
The set of features implemented by an architecture $C$ is defined as $Provided\_by_{\spl}(C) = \{f | implements_{\spl}(C, f)\}$.
\end{defn}

In ECPL, $implements_{\spl}(spec_2, Power\,lock)$ holds but $implements_{\spl}(spec_1, Power\,lock)$ does not hold. Moreover, if one considers $prov$ as given in Table \ref{tab_prov} without the last line, $implements_{\spl}(arch, f\_Speed)$ never holds for any architecture $arch$ because $prov(f\_Speed) = \emptyset$ even if $req(f\_Speed)=\{\{c\_Speed\}\}$. 

With the basic definitions above, we can now define when an architecture exactly implements a specification.
\begin{defn}[Realization]
 Given $C\in \arch$ and $F\in \spec$, $Realizes(C,F)$ if $F = Provided\_by(C)$.
\end{defn}
Due to the required equality, we have the following easy result.
\begin{prop}
An architecture realizes at most one specification in an SPL.
\end{prop}

The $realizes$ definition in the above imposes a strictness on the implementations. Thus, in the ECPL example, the architecture $arch_2$ realizes the specification $spec_2$, but it does not realize $spec_1$ even though it provides the implementation of all the features of $spec_1$. In many cases, this may be a practical definition. Hence, we relax the definition of realization in the following.

\begin{defn}[Covers]
Given $C\in \platform$ and $F\in \scope$, $C$ covers $F$ if $Provided\_by(C) \in \scope \mand F \subseteq Provided\_by(C)$.
\end{defn}
The additional condition ($Provided\_by(C) \in \scope$) is added to ensure that the chosen $C$ provides the implementation of a specification in the scope. In ECPL, $\components_2$ covers $\features_1$ but $\components_1$ does not cover (or even realize) anything.

Given  $F, F' \in \scope$, let $F \subset F'$, Then, $F'$ is called the extension of $F$. The following simple proposition establishes a connection between the relations {\em realizes} and {\em covers}. Figure~\ref{fig:cover-realize-extend} depicts these relations pictorially.

\begin{figure}[hbpt]
 \begin{center}
   {\includegraphics[scale=0.5]{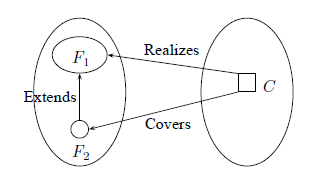}}
  \caption{Specification $F_1$ extends $F_2$, Architecture C realizes $F_1$ and covers $F_2$} 
  \label{fig:cover-realize-extend}
  \end{center}
\end{figure}

\begin{prop}
\label{prop:cover-realize-extension}
Given $C\in \platform$ and $F\in \scope$ and $C$ covers $F$. Then, there is an extension $F'$ of $F$ in $\scope$ such that $Realizes(C,F')$. Hence, if there is no extension of $F$ in $\scope$, then $Realizes(C,F)$.
\end{prop}

In the ECPL case study, $arch_2$ covers $spec_1$, $spec_2$ extends $spec_1$, and $arch_2$ realizes $spec_2$.

The set of products of the SPL are now defined as the specifications and the architectures implementing them through the traceability relation. 
\begin{defn}[SPL Products]
Given an SPL $\spl = \tuple{\scope,\arch, \trace}$, the products of the SPL denoted as $Prod(\spl) \defines \{\tuple{F,C} | Covers(F, C), F\in \scope, C \in \platform \}$
\end{defn}

In the ECPL, out of 8 specifications and 9 architectures, there are 11 products. Even if the architecture $arch_3=\{Door\,lock\,manager, Power\,lock, c\_Manual, Auto\,lock,$ $Transmission\,in\,park\}$ "covers" the specification \linebreak $\{Power\,lock, f\_Manual\}$, this pair is not a product because $Provided\_by(arch_3)$ is not in the scope $\scope$. This is because $arch_3$ provides features $f\_Manual$ and $Shift\,out\,of\,park$ which should be exclusive.

\subsection{SPL Level Properties}
Given an SPL $\tuple{\scope, \platform,  \trace}$, we define two important relationships between the scope (specification space) and platform (architecture, or implementation, space).

\subsubsection{Completeness}
An SPL $\tuple{\scope, \platform,  \trace}$ is complete if $\forall F\in \scope \cdot \exists C \in \platform \cdot Covers(C,F)$. 

The completeness property of the SPL determines if the platform for the SPL is adequate to provide implementation for all the specifications in its scope. 

The ECPL is complete. For illustration's sake, let us omit the last entry in Table \ref{tab_prov}. Then, none of the specifications which include the feature $f\_Speed$ is realizable because $f\_Speed$ cannot be derived from any component.

\subsubsection{Soundness}
An SPL $\tuple{\scope, \platform, \trace}$ is sound if $\forall C\in \platform\cdot \exists F \in \scope\cdot Covers(C,F)$.

The soundness property relates to the non-redundancy of the platform in an SPL. If the architectures (sets of components) are generated using certain rules or constraints, soundness stipulates that only those architectures which provide an implementation of some specification are generated.

The ECPL is not sound because, for example, the architecture $arch_1$ does not realize any specification (feature set). This is the case with all the architecture where $Power\,lock$ is absent. Now, let us assume that the component $Power\,lock$ is mandatory. The ECPL is still not sound because of $arch_3$ only. If $arch_3$ is omitted from the platform, the remaining ECPL become sound.

\subsubsection{Existentially Explicit}
Given an SPL, and a specification $F\in \scope$, it is called an existentially explicit specification in the SPL if there exists a $C \in \platform \suchthat Realizes(C,F)$.

In ECPL, $spec_1$ and $spec_2$ are existentially explicit. However, another specification $spec_3 = \langle Power\,lock,$
 $f\_Automatic,$ $Door\,lock,Shift\,out\,of\,park\rangle$ is not, because none of the architecture realizes a specification with $Door\,lock$ and without $Door\,relock$.

\subsubsection{Universally Explicit}
Given an SPL, and a specification $F\in \scope$, it is called a universally explicit specification in the SPL if (i) there exists a $C \in \platform \suchthat Realizes(C,F)$ and (ii) for all $C \in \platform \suchthat Covers(C,F) \implies Realizes(C,F)$.

In ECPL, $spec_2$ is universally explicit. $spec_1$ is existentially explicit but not universally explicit because it is covered but not realized by the $arch_2$.

It follows from Proposition~\ref{prop:cover-realize-extension} that
\begin{prop}
If $F\in \scope$ is covered by some architecture but is not extendable, then it is universally explicit. If $F$ is universally explicit, then none of its extensions has a covering architecture.
\end{prop}

In the ECPL, $spec_2$ is covered and cannot be extended; so it is universally explicit. On the contrary, if a specification has an extension which is covered, the same also covers the extended specification.

\subsubsection{Unique Implementation}
A given specification may be implemented by multiple architectures. This may be a desirable criterion of the platform from the perspective of optimization among various choices. Thus the specifications which are implemented by single architectures are to be identified.

$F \in \scope$ has a unique implementation if $\exists C\in \platform \suchthat (Covers(C,F) \mand \forall C' \in \platform \suchthat (Covers(C', F) \implies C = C'))$.

In ECPL, each specifications including $Door\,relock$ has a unique implementation. On contrary, $spec_1$ has more than one implementation.

\subsubsection{Common, live and dead elements}
Identification of common, live and dead elements in an SPL is one of the basic analyses identified in the SPL community. We redefine these concepts in terms of the our notion of products.
\begin{enumerate}
    \item An element $e$ is {\em common} if $\forall \tuple{F,C} \in Prod(\spl) \cdot e \in F\cup C$.
	\item An element $e$ is {\em live} if $\exists \tuple{F,C} \in Prod(\spl) \cdot e \in F\cup C$.
	\item An element $e$ is {\em dead} if $\forall \tuple{F,C} \in Prod(\spl) \cdot e \not\in F\cup C$.
\end{enumerate}

In ECPL, the feature {\em Manual\,lock} is dead. All the other features are live. The component {\em Door\,lock\,manager} is common. 

\subsubsection{Superfluous Component}
A component is superfluous if the platform without the component suffices to provide the same set of specifications. 

Let $P \subseteq Prod(\spl)$, $spec(P) = \{F | \tuple{F,C} \in P\}$.
Let $Prod_{\neg c}(\spl) = \{\tuple{F,C} | \tuple{F,C} \in Prod(\spl) \mand (c \not\in C) \}$. $c$ is Superfluous if $spec(Prod(\spl)) = spec(Prod_{\neg c}(\spl))$.

Superfluousness is relative to a given platform. If in an SPL $\spl$, $prov(f) = \{\{a\}, \{b\}\}$, $\scope = \{\{f\}\}$ and $\platform = \{\{a\}, \{b\}\}$, then both $a$ and $b$ are superfluous w.r.t. $\spl$, whereas if either $\{a\}$ or $\{b\}$ is removed from the platform, the remaining $\{b\}$ or $\{a\}$ is not superfluous anymore (w.r.t. the reduced SPL).

\begin{lem}
Let $c \in \components$ be Superfluous for $\spl$. Then, for every $C\in \platform (c \in C \implies (\exists C'\in \platform \cdot c\not\in C' \mand Provided\_by(C) = Provided\_by(C'))$.
\end{lem}


\subsubsection{Redundant Component}
A component is redundant if it is not contributing to any feature in any architecture in the platform. $c \in \components$ is redundant if for every $C\in \platform (c \in C \implies (\exists C'\in \platform \cdot (c\not\in C' \mand C' \subseteq C \mand Provided\_by(C) = Provided\_by(C'))$.

Note that redundancy is a stronger version of superfluousness; a redundant component is superfluous whereas a superfluous element many not be redundant.

In ECPL, no component is neither superfluous nor redundant. Let us assume that we have a component called $Door\,Relock_{Alt}$ such that $\{Door\,Relock_{Alt}, \,Auto\,lock \}$ provides the feature $Door\,Relock$. This component would be redundant because $Auto\,lock$ already provides the feature $Door\,Relock$.

It is expected that an SPL can be optimized by omitting the redundant components without affecting the set of products.
\begin{lem}
Let $c\in \components$ be redundant. Construct a SPL $\spl' = \tuple{\scope, \trace', \platform}$ where, $\trace'$ be a traceability relation with $req'(f) = req(f) \setminus \{C | c\in C\}$ and $prov'(f) = prov(f) \setminus \{C | c\in C\}$. Then, $Prod(\spl) = Prod(\spl')$.
\end{lem}

\subsubsection{Critical Component}
Given an $f \in \features$, a component $c$ is critical for $f$ if for all $C\in \platform, (c \not\in C \implies \neg implements_{\spl}(C, f))$.

In ECPL, all the components are critical. Let us assume a component $Auto\,lock_{Alt}$ which is an alternative to $Auto\,lock$ and also provides the feature $Door\,lock$. In such case neither $Auto\,lock$ or $Auto\,lock_{Alt}$ are critical for the feature $Door\,lock$ but $Auto\,lock$ remains critical for the feature $Auto\,relock$.

\subsubsection{Emerging Features}
When a specification is not realizable, but is covered by one or more architectures, the emerging features $Emerging(F) \defines \{ \tuple{C, Provided\_by(C)\setminus F} | Covers(C,F)\}$. 

$Emerging(F)$ gives the covering architectures and the emerging features corresponding to the architecture. 

In ECPL, while considering the only architecture that cover $\tuple{Power\,lock,Manual,Door\,lock,f\_Speed}$, $Door\,relock$ will emerge.
 


\subsection{Canonical Traceability Relation}
A given traceability relation can be reduced to a canonical form without affecting the set of features implementable in the SPL. We define the canonical form in the following.

\begin{defn}
$\trace$ is {\em non-redundant} if for every feature $f$,
\begin{enumerate}
    \item $C_i, C_j \in prov(f), i\neq j$ implies $C_i \not\subseteq C_j$, and
	\item $C_i, C_j \in req(f), i\neq j$ implies $C_i \not\subseteq C_j$.
\end{enumerate} 
\end{defn}
 Intuitively, if a smaller set of components implements a feature, a larger set also will. On the other hand, if a larger set of components is required to implement a feature, a smaller set is required automatically. Given a traceability relation, one can check if it is non-redundant and convert it to a non-redundant relation by removing the larger (resp. smaller) sets in $prov(f)$ (resp. $req(f)$).

\begin{defn}
$\trace$ is {\em internally consistent} if $\forall f \in \features$, $\forall C \subseteq \components$, $(C \in prov(f) \implies (\exists C' \in req(f) \cdot C' \subseteq C))$.
\end{defn}
Intuitively, internal consistency of a traceability relation states that each set of components in $prov(f)$ can indeed satisfy the mandatory requirements (coming from $req(f)$ of $f$.

Given a traceability relation, we can reduce it to a canonical form by the following  operations for the $prov(.)$ and $req(.)$ of each feature $f$.
%

\begin{algorithm}
\caption{Canonization of Traceability Relation}
\begin{algorithmic}[1]
\IF
{$prov(f) = \emptyset$ or $prov(f)$ is undefined} 
\STATE $prov(f) \leftarrow \bot; req(f)\leftarrow \bot$
\ENDIF
\IF {$C_i, C_j \in prov(f), i\neq j$, $C_i \subseteq C_j$}
\STATE $prov(f)\leftarrow prov(f) \setminus \{C_j\}$.
\ENDIF
\IF {$C_i, C_j \in req(f), i\neq j$, $C_i \subseteq C_j$} 
\STATE $req(f) \leftarrow req(f) \setminus \{C_i\}$.
\ENDIF
\IF {$C \in prov(f)$, but $forall C_i \in req(f), C_i \not\subseteq C$}
 \STATE $prov(f) \leftarrow prov(f) \setminus \{C\}$.  
\ENDIF
\end{algorithmic}
\end{algorithm}

\begin{clm}
For a given SPL $\spl = \tuple{\scope, \platform, \trace}$, the above procedure results in a canonical traceability relation $\trace'$ such that for all $C \subseteq \components$, $implements_{\spl}(C,f)$ iff $implements_{\spl'}(C,f)$.
\end{clm}
\begin{proof}
The canonization algorithm stops when no rules are applicable. Then the conditions of the rules ensure that the resulting traceability relation is canonical.

In order to prove the preservation of implementability, it is easy to show that each rule preserves implementability. 
\end{proof}

\begin{thm}
\label{th:new-implements}
If $\spl$ is an SPL with a canonical traceability relation, $implements_{\spl}(C, f)$ if $\exists C_1 \in prov(f) \suchthat C_1 \subseteq C$.
\end{thm}
\begin{proof}
In a canonical traceability relation, due to internal consistency, for every $C' \in prov(f), \exists C'' \in req(f) \suchthat C'' \subseteq C'$. Hence the result.
\end{proof}

Since one can always canonize the traceability relation of an SPL, henceforth we will assume that the SPL under scope is canonical. Thereby, the definition of implementation will henceforth be as given in \ref{th:new-implements}.

\section{Analysis of the ECPL}
\label{sec_casestudy1}
In this section, we analyze some properties of the ECPL example using QuBE. 

In ECPL, there are total 8 Features and 13 Components. The features are listed in Table ~\ref{tab_features} and the components are given in Table ~\ref{tab_components}. 

\begin{table}
	\centering
	\begin{tabular}{|c|c|}
	\hline
	$\textbf{Short-Hand}$ &   $\textbf{Feature}$  \\	
	\hline
	\hline
	$F_1$ &   $Manual$ $Lock$  \\
	\hline
	$F_2$ &   $Power$ $Lock$   \\
	\hline
	$F_3$ &   $Door$ $Lock$   \\
	\hline
	$F_4$ &   $Door$ $Relock$   \\
	\hline
	$F_5$ &   $F\_automatic$ \\
	\hline
	$F_6$ &   $F\_manual$ \\
	\hline
	$F_7$ &   $F\_speed$   \\
	\hline
	$F_8$ &   $Shift$ $out$ $of$ $Park$ \\
	\hline
	\end{tabular}
	\caption{Features in ECPL.}
    \label{tab_features}
\end{table}

\begin{table}
	\centering
	\begin{tabular}{|c|c|}
	\hline
	$\textbf{Short-Hand}$ &   $\textbf{Component}$  \\	
	\hline
	\hline
	$C_1$ &   $Door$ $Lock$ $Manager$  \\
	\hline
	$C_2$ &   $Unlock$ $Driver$ $Door$   \\
	\hline
	$C_3$ &   $Unlock$ $all$ $doors$   \\
	\hline
	$C_4$ &   $Lock$ $all$ $doors$  \\
	\hline
	$C_5$ &   $Auto$ $Lock$ \\
	\hline
	$C_6$ &   $Power$ $Lock$  \\
	\hline
	$C_7$ &   $Courtesy$ $switch$  \\
	\hline
	$C_8$ &   $Key$ $signal$ \\
	\hline
	$C_9$ &   $Sill$ $door$ $signal$  \\
	\hline
	$C_{10}$ &  $C\_automatic$ \\
	\hline
	$C_{11}$ &   $C\_manual$ \\
	\hline
	$C_{12}$ &   $Gear$ $in$ $park$  \\
	\hline
	$C_{13}$ &  $C\_speed$ \\
	\hline	
	\end{tabular}
	\caption{Components in ECPL.}
	\label{tab_components}
\end{table}

A specification is a subset of Features $\features$. The scope of an SPL is a collection of specifications: $\spec \subseteq \powerset(\features)$. In our example, scope of ECPL is $\spec = \{$ $S_1,$ $S_2,$ $S_3,$ $S_4,$ $S_5,$ $S_6,$ $S_7,$ $S_8$$\}$. All the specifications are represented in tabular form as shown in Table ~\ref{tab_spec}. A specification corresponds to a column and the 1's in the column select the features in the specification.

\begin{enumerate}
\item $S_1=\{$$Power$ $Lock$$,$ $F\_automatic$$\}$
\item $S_2=\{$$Power$ $Lock$$,$ $F\_manual$$\}$
\item $S_3=\{$$Power$ $Lock$$,$ $F\_automatic$$,$ $Door$ $Lock$$,$ $F\_speed$$\}$
\item $S_4=\{$$Power$ $Lock$$,$ $F\_manual$$,$ $Door$ $Lock$$,$ $F\_speed$$\}$
\item $S_5=\{$$Power$ $Lock$$,$ $F\_automatic$$,$ $Door$ $Lock$$,$ $Shift$ $out$ $of$ $Park$$\}$
\item $S_6=\{$$Power$ $Lock$$,$ $F\_automatic$$,$ $Door$ $Lock$$,$ $F\_speed$$,$ $Door$ $relock$$\}$.
\item $S_7=\{$$Power$ $Lock$$,$ $F\_manual$$,$ $Door$ $Lock$$,$ $F\_speed$$,$ $Door$ $relock$$\}$.
\item $S_8=\{$$Power$ $Lock$$,$ $F\_automatic$$,$ $Door$ $Lock$$,$ $Shift$ $out$ $of$ $Park$$,$ $Door$ $relock$$\}$.
\end{enumerate}

An architecture is a subset of components $\components$. An SPL platform consists of a set of architectures: $\platform \subseteq \powerset(\components)$. In ECPL, the platform is $\platform = \{$$A_1,$ $A_2,$ $A_3,$ $A_4,$ $A_5,$ $A_6,$ $A_7,$ $A_8,$ $A_9$$\}$. The architectures are represented in Table ~\ref{tab_arch}.

\begin{enumerate}
\item $A_1=\{$$Door$ $Lock$ $Manager$$,$ $Unlock$ $Driver$ $Door$$,$ $Unlock$ $all$ $doors$$,$ $Lock$ $all$ $doors$$\}$
\item $A_2=\{$$Door$ $Lock$ $Manager$$,$ $Unlock$ $Driver$ $Door$$,$ $Unlock$ $all$ $doors$$,$ $Lock$ $all$ $doors$$,$ $Auto$ $Lock$$,$ $C\_speed$$\}$
\item $A_3=\{$$Door$ $Lock$ $Manager$$,$ $Unlock$ $Driver$ $Door$$,$ $Unlock$ $all$ $doors$$,$ $Lock$ $all$ $doors$$,$ $Auto$ $Lock$$,$ $Gear$ $in$ $park$$\}$
\item $A_4=\{$$Door$ $Lock$ $Manager$$,$ $Unlock$ $Driver$ $Door$$,$ $Unlock$ $all$ $doors$$,$ $Lock$ $all$ $doors$$,$ $Power$ $Lock$$,$ $Courtesy$ $switch$$,$ $Key$ $signal$$,$ $Sill$ $door$ $signal$$,$
$C\_automatic$$\}$
\item $A_5=\{$$Door$ $Lock$ $Manager$$,$ $Unlock$ $Driver$ $Door$$,$ $Unlock$ $all$ $doors$$,$ $Lock$ $all$ $doors$$,$ $Power$ $Lock$$,$ $Courtesy$ $switch$$,$ $Key$ $signal$$,$ $Sill$ $door$ $signal$$,$
$C\_manual$$\}$
\item $A_6=\{$$Door$ $Lock$ $Manager$$,$ $Unlock$ $Driver$ $Door$$,$ $Unlock$ $all$ $doors$$,$ $Lock$ $all$ $doors$$,$ $Auto$ $Lock$$,$ $C\_speed$$,$ $Power$ $Lock$$,$ $Courtesy$ $switch$$,$ $Key$ $signal$$,$ $Sill$ $door$ $signal$$,$ $C\_automatic$$\}$
\item $A_7=\{$$Door$ $Lock$ $Manager$$,$ $Unlock$ $Driver$ $Door$$,$ $Unlock$ $all$ $doors$$,$ $Lock$ $all$ $doors$$,$ $Auto$ $Lock$$,$ $C\_speed$$,$ $Power$ $Lock$$,$ $Courtesy$ $switch$$,$ $Key$ $signal$$,$ $Sill$ $door$ $signal$$,$ $C\_manual$$\}$
\item $A_8=\{$$Door$ $Lock$ $Manager$$,$ $Unlock$ $Driver$ $Door$$,$ $Unlock$ $all$ $doors$$,$ $Lock$ $all$ $doors$$,$ $Auto$ $Lock$$,$ $Gear$ $in$ $park$$,$ $Power$ $Lock$$,$ $Courtesy$ $switch$$,$ $Key$
$signal$$,$ $Sill$ $door$ $signal$$,$ $C\_automatic$$\}$
\item $A_9=\{$$Door$ $Lock$ $Manager$$,$ $Unlock$ $Driver$ $Door$$,$ $Unlock$ $all$ $doors$$,$ $Lock$ $all$ $doors$$,$ $Auto$ $Lock$$,$ $Gear$ $in$ $park$$,$ $Power$ $Lock$$,$ $Courtesy$ $switch$$,$ $Key$
$signal$$,$ $Sill$ $door$ $signal$$,$ $C\_manual$$\}$
\end{enumerate}

\begin{table}
	\centering
	\begin{tabular}{|c|c|c|c|c|c|c|c|c|}
	\hline
	$\textbf{Specifications} \over{\textbf{Features}} $ & $S_1$ & $S_2$ & $S_3$ & $S_4$ & $S_5$ & $S_6$ & $S_7$ & $S_8$\\
	\hline
	\hline	
	$F_1$ &   &  &  &  &  &  &  &  \\
	\hline
	$F_2$ & 1 & 1 & 1 & 1 & 1 & 1 & 1 & 1 \\
	\hline
	$F_3$ &   &   & 1 & 1 & 1 & 1 & 1 & 1 \\
	\hline
	$F_4$ &   &   &   &   &   & 1 & 1 & 1 \\
	\hline
	$F_5$ & 1 &   & 1 &   & 1 &   &   &   \\
	\hline
	$F_6$ &   & 1 &   & 1 &   & 1 &   & 1 \\
	\hline
	$F_7$ &   &   & 1 & 1 &   & 1 & 1 &   \\
	\hline
	$F_8$ &   &   &   &   & 1 &   &   & 1 \\
	\hline
	\end{tabular}
	\caption{Specifications in tabular form.}
	\label{tab_spec}
\end{table}

\begin{table}
	\centering
	\begin{tabular}{|c|c|c|c|c|c|c|c|c|c|}
	\hline 
	$\textbf{Architectures} \over{\textbf{Components}}$ & $A_1$ & $A_2$ & $A_3$ & $A_4$ & $A_5$ & $A_6$ & $A_7$ & $A_8$ & $A_9$ \\
	\hline
	\hline	
	$C_1$ & 1 & 1 & 1 & 1 & 1 & 1 & 1 & 1 & 1 \\
	\hline
	$C_2$ & 1 & 1 & 1 & 1 & 1 & 1 & 1 & 1 & 1 \\
	\hline		
	$C_3$ & 1 & 1 & 1 & 1 & 1 & 1 & 1 & 1 & 1 \\
	\hline	
	$C_4$ & 1 & 1 & 1 & 1 & 1 & 1 & 1 & 1 & 1 \\
	\hline		
	$C_5$ &   & 1 & 1 &   &   & 1 & 1 & 1 & 1 \\
	\hline			
	$C_6$ &   &   &   & 1 & 1 & 1 & 1 & 1 & 1 \\
	\hline
	$C_7$ &   &   &   & 1 & 1 & 1 & 1 & 1 & 1 \\
	\hline
	$C_8$ &   &   &   & 1 & 1 & 1 & 1 & 1 & 1 \\
	\hline		
	$C_9$ &   &   &   & 1 & 1 & 1 & 1 & 1 & 1 \\
	\hline		
	$C_{10}$ &   &   &   & 1 &   & 1 &   & 1 &   \\
	\hline		
	$C_{11}$ &   &   &   &   & 1 &   & 1 &   & 1 \\
	\hline		
	$C_{12}$ &   &   & 1 &   &   &   &   & 1 & 1 \\
	\hline		
	$C_{13}$ &   & 1 &   &   &   & 1 & 1 &   &   \\				
	\hline					
	\end{tabular}
	\caption{Architectures in tabular form.}
	\label{tab_arch}
\end{table}

The traceability relations (provides and requires) are as in Tables ~\ref{tab_req} and ~\ref{tab_prov}. We reproduce the tables here for ease of reference.

\begin{table}[h!]
	\centering
	\begin{tabular}{|c|c|c|}
	\hline
	$\textbf{Feature}$ & $\textbf{Component}$ \\
	\hline
	\hline
	Power lock& Door lock manager\& Power lock\\
	Door lock& Auto lock\\
	Door relock& Auto lock\\
	F\_automatic& C\_automatic\\
	F\_manual& C\_manual\\
	Shift out of Park& Gear in Park\\
	F\_speed& C\_speed\\
	\hline
	\end{tabular}
	\caption{Requires relation in ECPL}
\end{table}

\begin{table}[h!]
	\centering
	\begin{tabular}{|c|c|c|}
	\hline
	$\textbf{Component/Interconnection}$ & $\textbf{Feature}$\\
	\hline
	\hline
	 Door lock manager \& Power lock& Power lock \\
	 C\_automatic& F\_automatic\\
	 C\_manual& F\_manual\\
	 Auto lock & Door lock \& Door relock \\
	 Gear in park  & Shift out of park \\
	 \hline
	 C\_speed & F\_speed\\
	 \hline	
	\end{tabular}
	\caption{Provides relation in ECPL}
\end{table}

\paragraph{Implements:} 
$implements_{\spl}(A,$ $f)$ if $\exists C_1 \in prov(f), C_2 \in req(f) \suchthat C_2 \subseteq C_1 \subseteq A$. The set of features implemented by an architecture $A$ is defined as $Provided\_by_{\spl}(A) = \{f | implements_{\spl}(A,$ $f)\}$.

$Examples:$ In ECPL, check if $implements_{\spl}($$A_4$$,$ $Power$ $Lock)$ holds.

$Solution:$ Let P1 = $prov(Power$ $Lock)$. From Table ~\ref{tab_prov}, P1 =  $prov(Power$ $Lock) =\{\{$$Door$ $Lock$ $Manager,$ $Power$ $Lock$$\}\}$. Let R1 = $req(Power$ $Lock)$. From Table ~\ref{tab_req}, R1 = $req(Power$ $Lock)=\{\{$$Door$ $Lock$ $Manager,$ $Power$ $Lock$$\}\}$. Since $R_1$ $\subseteq P_1$ $\subseteq A_4$, $implements_{\spl}($$A_4$$,$ $Power$ $Lock)$ holds. On other hand $R_1$ $\subseteq P_1$ $\nsubseteq A_1$, hence$implements_{\spl}($$A_1$$,$ $Power$ $Lock)$ does not hold.

For each feature, we can find the architectures which implement it. The results are listed in Table ~\ref{tab_impl}: the 1's in the column corresponding to an architecture gives us the features implemented.

\begin{table}
\centering
\begin{tabular}{|c|c|c|c|c|c|c|c|c|c|}
\hline $\textbf{Architectures} \over{\textbf{Features}}$ & $A_1$ & $A_2$ & $A_3$ & $A_4$ & $A_5$ & $A_6$ & $A_7$ & $A_8$ & $A_9$ \\
\hline
\hline $F_1$ &   &   &   &   &   &   &   &   &   \\ 
\hline $F_2$ &   &   &   & 1 & 1 & 1 & 1 & 1 & 1 \\ 
\hline $F_3$ &   & 1 & 1 &   &   & 1 & 1 & 1 & 1 \\ 
\hline $F_4$ &   & 1 & 1 &   &   & 1 & 1 & 1 & 1 \\ 
\hline $F_5$ &   &   &   & 1 &   & 1 &   & 1 &   \\ 
\hline $F_6$ &   &   &   &   & 1 &   & 1 &   & 1 \\ 
\hline $F_7$ &   & 1 &   &   &   & 1 & 1 &   &   \\ 
\hline $F_8$ &   &   & 1 &   &   &   &   & 1 & 1 \\ 
\hline 
\end{tabular} 
	\caption{Feature implementation in given SPL.}
	\label{tab_impl}
\end{table}

\paragraph{Realization:} Given $A\in \arch$ and $S\in \spec$, $Realizes(A,$ $S)$ if $S = Provided\_by(A)$.

$Example:$ In ECPL, check if $Realizes(A_4,$ $S_1)$ holds.

$Solution:$ The specification $S_1$ has the features $\{$$Power$ $Lock$$,$ $F\_automatic$$\}$. From Table ~\ref{tab_impl}, $Provided\_by(A_4)=\{$$Power$ $Lock$$,$ $F\_automatic$$\}$. Since $Provided\_by(A_4)$ = $S_1$, $Realizes(A_4,$ $S_1)$ holds. On the other hand, \\
$Provided\_by(A_5)=\{$$Power$ $Lock$$,$ $F\_manual$$\}$ $\neq S_1$, hence $Realizes(A_5,$ $S_1)$ does not hold.

The Table \ref{tab_real} shows all the specifications and it's corresponding realized architectures.

\begin{table}
\centering
\begin{tabular}{|c|c|c|c|c|c|c|c|c|c|}
\hline $\textbf{Architectures}\over{\textbf{Specifications}}$ & $A_1$ & $A_2$ & $A_3$ & $A_4$ & $A_5$ & $A_6$ & $A_7$ & $A_8$ & $A_9$ \\
\hline
\hline $S_1$ &   &   &   & 1 &   &   &   &   &   \\ 
\hline $S_2$ &   &   &   &   & 1 &   &   &   &   \\ 
\hline $S_3$ &   &   &   &   &   &   &   &   &   \\ 
\hline $S_4$ &   &   &   &   &   &   &   &   &   \\ 
\hline $S_5$ &   &   &   &   &   &   &   &   &   \\ 
\hline $S_6$ &   &   &   &   &   & 1 &   &   &   \\ 
\hline $S_7$ &   &   &   &   &   &   & 1 &   &   \\ 
\hline $S_8$ &   &   &   &   &   &   &   & 1 &   \\ 
\hline 
\end{tabular} 
	\caption{Specifications and the realizing architectures.}
	\label{tab_real}
\end{table}

\paragraph{Covers:} Given $A\in \platform$ and $S\in \scope$, $A$ covers $S$ if $Provided\_by(A) \in \scope \mand S \subseteq Provided\_by(A)$.

$Example:$ In ECPL, check $Covers(A_6,$ $S_1)$ Hold?

$Solution:$ The specification $S_1$ has $\{$$Power$ $Lock$$,$ $F\_automatic$$\}$ features. From Table ~\ref{tab_impl}, $Provided\_by(A_6)=\{$$Power$ $Lock$$,$ $Door$ $Lock$$,$ $Door$ $Relock$$,$ $F\_automatic$$\}$. Since $Provided\_by(A_6)$ $\in \scope$ and $S_1 \subseteq Provided\_by(A_6)$, hence $Covers(A_6,$ $S_1)$ hold. On the other hand, $Provided\_by(A_5)=\{$$Power$ $Lock$$,$ $F\_manual$$\}$ $\in \scope$ but $S_1 \nsubseteq Provided\_by(A_5)$, hence $Covers(A_5,$ $S_1)$ does not hold.

Similarly, for all other specifications we can find the architectures which cover the specifications. The Table \ref{tab_cove} has all the specifications and their covering architectures.

\begin{table}
\centering
\begin{tabular}{|c|c|c|c|c|c|c|c|c|c|}
\hline $\textbf{Architectures}\over{\textbf{Specifications}}$ & $A_1$ & $A_2$ & $A_3$ & $A_4$ & $A_5$ & $A_6$ & $A_7$ & $A_8$ & $A_9$ \\
\hline
\hline $S_1$ &   &   &   & 1 &   & 1 &   & 1 &   \\ 
\hline $S_2$ &   &   &   &   & 1 &   & 1 &   &   \\ 
\hline $S_3$ &   &   &   &   &   & 1 &   &   &   \\ 
\hline $S_4$ &   &   &   &   &   &   & 1 &   &   \\ 
\hline $S_5$ &   &   &   &   &   &   &   & 1 &   \\ 
\hline $S_6$ &   &   &   &   &   & 1 &   &   &   \\ 
\hline $S_7$ &   &   &   &   &   &   & 1 &   &   \\ 
\hline $S_8$ &   &   &   &   &   &   &   & 1 &   \\ 
\hline 
\end{tabular} 
	\caption{Specifications and their covering architectures.}
	\label{tab_cove}
\end{table}

\subsection{SPL Level Properties of ECPL}

\paragraph{Completeness:} In ECPL, from Table ~\ref{tab_cove} one can observe that every specification in scope $\scope$ is covered by some architecture in platform $\platform$. Hence, ECPL is complete.

\paragraph{Soundness:} From Table ~\ref{tab_cove} one can observe that the architectures $S_1,$ $S_2$ and $S_3$ do not cover any specification in scope $\scope$. Hence, ECPL is not sound.

\paragraph{Existentially Explicit:}
It is observed from Table ~\ref{tab_real} that the architectures $S_1$, $S_2$, $S_6$, $S_7$ and $S_8$ are realized by the architectures $A_4$, $A_5$, $A_6$, $A_7$ and $A_8$ respectively. Hence these specifications are existentially explicit. From the same table, one can observe that the specifications $S_3$, $S_4$ and $S_5$ are not realized by any architecture in the given platform.

\paragraph{Universally Explicit:} In ECPL, from Table ~\ref{tab_real} and ~\ref{tab_cove}, it is observed that the specifications $S_6$, $S_7$ and $S_8$  are realized by the architectures $A_6$, $A_7$ and $A_8$ respectively, and these are the only architectures which cover the respective specifications. Hence, these specifications are universally explicit. As we have already seen from Table ~\ref{tab_real}, the architectures $S_3$, $S_4$ and $S_5$ are not realized at all. The remaining architectures $S_1$ and $S_2$ are realized by $A_4$ and $A_5$ respectively, but $S_1$ is also strictly covered (covered but not realized) by architectures $A_6$ and $A_7$ and $S_2$ is strictly covered by $A_7$. Hence, the specifications $S_1$, $S_2$, $S_3$, $S_4$ and $S_5$ are not universally explicit.

\paragraph{Unique Implementation:}
In an SPL, a given specification is said to be uniquely implemented if it is covered by exactly one architecture. In ECPL, from Table ~\ref{tab_cove} it is found that the specifications $S_3$, $S_4$, $S_5$, $S_6$, $S_7$ and $S_8$ are covered by exactly one architecture ($A_6$, $A_7$, $A_8$, $A_6$, $A_7$, $A_8$ respectively). Hence, these specifications have unique implementation. On the other hand, the specifications $S_1$ and $S_2$ have multiple implementations.

\paragraph{Products:}
In ECPL, from Table ~\ref{tab_cove} we get $Prod(\spl) =$$\{$$\tuple{S_1,A_4},$ $\tuple{S_1,A_6},$ $\tuple{S_1,A_8},$ $\tuple{S_2,A_5},$ $\tuple{S_2,A_7},$ $\tuple{S_3,A_6},$ $\tuple{S_4,A_7},$ $\tuple{S_5,A_8},$ $\tuple{S_6,A_6},$ $\tuple{S_7,A_7},$ $\tuple{S_8,A_8}$$\}$.

\paragraph{Common, live and dead elements:}
From the set of products and referring to the tables ~\ref{tab_spec} and ~\ref{tab_arch}, we find that the common elements of ECPL are $\{$$Power$ $Lock^{1}$$,$ $Door$ $Lock$ $Manager$$,$ $Unlock$ $Driver$ $Door$$,$ $Unlock$ $all$ $doors$$,$ $Lock$ $all$ $doors$$,$ $Power$ $Lock^{2}$$,$ $Courtesy$ $switch$$,$ $Key$ $signal$$,$  $Sill$ $door$ $signal$$\}$. $Power$ $Lock^{1}$ is the feature and $Power$ $Lock^{2}$ is the component.

The live elements for $Prod(\spl)$ are $\{$$Power$ $Lock^{1}$$,$ $Door$ $Lock$$,$ $Door$ $Relock$$,$ $F\_automatic$$,$ $F\_manual$$,$ $F\_speed$$,$ $Shift$ $out$ $of$ $Park$ $,$ $Door$ $Lock$ $Manager$$,$ $Unlock$ $Driver$ $Door$$,$ $Unlock$ $all$ $doors$$,$ $Lock$ $all$ $doors$$,$ $Auto$ $Lock$$,$ $Power$ $Lock^{2}$$,$ $Courtesy$ $switch$$,$ $Key$ $signal$$,$ $Sill$ $door$ $signal$$,$ $C\_automatic$$,$ $C\_manual$$,$ $Gear$ $in$ $park$$,$ $C\_speed$$\}$. The only dead element is $Manual$ $Lock$.

\paragraph{Superfluous Component:} 
There are no superfluous components in ECPL. For example, consider the element $Auto Lock$. The specification $S_1$ is covered by architectures $A_4$, $A_6$ and $A_8$. If architectures $A_6$ and $A_8$, which include $Auto Lock$, are removed, then $S_1$ is still in the product (being implemented by $A_4$). However, $A_6$ is the only architecture covering $S_3$. Hence, when $A_6$ is removed, $\tuple{S_3, A_6}$ is removed from the list of products. This implies that $Auto Lock$ is not superfluous.

\paragraph{Redundant Component:} A component is redundant if it is not contributing to any feature in any architecture in the platform.
In ECPL, there are not any redundant component. Let us assume we have a component called $Door\,Relock_{Alt}$ such that $\{Door\,Relock_{Alt},$ $Auto\,Lock`\}$ provides the feature $Door\,Relock$. This component would be redundant because $Auto\,lock$ already provide the feature $Door\,Relock$.

\paragraph{Critical Component:} 
In ECPL, all the components are critical. Let us remove the component $C\_automatic$ from architecture $A_4$. Then, $implements_{\spl}(A_4,$ $F\_automatic))$ will not hold. Hence, we can say that the component $C\_automatic$ is critical for feature $F\_automatic$.

\paragraph{Emerging Features:} 
In ECPL, the specification $S_4$ is not realized by any architecture but it is covered by $A_7$. So the set of emerging features is $Provided\_by(A_7)$ $-$ $S_4$$=$$\{$$Door$ $relock$$\}$.

\subsection{Performance}
We have recorded the time required to check the satisfiability of the formulae for some analysis problems using QuBE (Refer Table~\ref{tab_time}). Each formula has been run three times and the average time is calculated. The performance of QuBE  seems quite good for small SPLs the size of ECPL.

\begin{table}
\centering
\begin{tabular}{|c|c|c|c|c|}
\hline $\textbf{Properties and Formulae}$ & $\textbf{Test}$ $\textbf{1}$ & $\textbf{Test}$ $\textbf{2}$ & $\textbf{Test}$ $\textbf{3}$ & $\textbf{Average Time(ms)}$\\
\hline
\hline $Implements$ &  3 & 2  &  2 & 2.33 \\ 
\hline $realizes$ &  2 & 2  &  2 & 2 \\ 
\hline $covers$ &  3 & 2  &  2 & 2.33 \\ 
\hline $complete$ &  3 & 2  &  2 & 2.33 \\  
\hline $sound$ &  4 & 3  &  3 & 3.33 \\  
\hline $existentially$ $explicit$ &  3 & 2  &  3 & 2.67 \\ 
\hline $critical$ &  3 & 3  &  3 & 3 \\ 
\hline $extended$ $features$ &  2 & 2  &  2 & 2 \\ 
\hline 
\end{tabular} 
	\caption{Time complexity for Properties and Formulae}
	\label{tab_time}
\end{table}

\section{Analysis between the specification and the implementation perspectives}
\label{sec_encoding}
In the literature, different analysis problems in SPL are usually encoded as propositional satisfiability problems\cite{Batory2005} and SAT solvers such as Yices, Bddsolve\cite{bddsolve} etc. are used to solve the problems. However, looking at the definition of $implements$ and the subsequent problems, we observe that there is quantification over the features and components which can be encoded as propositions. In fact, we show in the following that it is possible to transform the analysis problems of the previous section into QBF formula such that  the questions have an affirmative answer iff the corresponding QBF formulae hold.

\begin{enumerate}
\item Let ${\cal C}=\{c_1, \dots, c_n\}$ be the set of all components and let ${\cal F}=\{f_1, \dots, f_m\}$ be the set of all features. A subset of ${\cal F}$ is a specification, while a subset of ${\cal C}$ is called an architecture. A platform is a set of architectures $\overline{\cal C} \subseteq {\cal P}({\cal C})$. 
A scope is a set of specifications $\overline{\cal F} \subseteq {\cal P}({\cal F})$.


 \item  Given an architecture $C=\{c_1, \dots, c_k\}$, let $Prop(C)$ be the tuple of propositions

$Prop(C)(i) = \left\{  \begin{array}{cc} c_i~\mbox{if}~c_i \in C\\
                                        \neg c_i ~\mbox{if}~c_i \notin C
                       \end{array} \right.$

Thus, $Prop(C)$ is an $n$-tuple made up of 0's and 1's. 
The tuple $Prop(F)$ for a specification $F$ can be defined similarly. 





\item Let $f$ be a feature. Let $prov(f)=\{S_1, S_2 \dots, S_k\}$. Each $S_j$ is a set of components that provides $f$. Then we define $formula \_ prov(f)$  as $\bigvee_j \bigwedge_{c_i \in S_j} c_i$. 
$formula\_prov(f)$ is satisfiable whenever there is some set $S_j$ of components that
provide feature $f$. If the set $prov(f)$ is undefined(empty), then 
  $formula\_prov(f)$ is FALSE, since there are no components that provide feature $f$.

\item Let $f$ be a feature. Let $req(f)=\{S_1, S_2 \dots, S_k\}$.
$f$ requires at least one set $S_j$ of components for its implementation.
Then, we define $formula \_req(f)=\bigvee_j \bigwedge_{c_i \in S_j} c_i$. 
$formula\_req(f)$ is satisfiable iff $req(f)$ has at least one set (say $S_j$) of its required components. If 
$req(f)$ is empty or undefined, then $formula \_ req(f)$ is TRUE, since there are no requirements 
for $f$.


\item Let $f$ be a feature and let $prov(f)=\{S_1, S_2 \dots, S_k\}$. 
Given a tuple of component parameters $(c'_1, \dots, c'_n)$ where each $c'_i$ is 0 or 1, 
and a feature $f$,
we define the formula $f \_ implements(c'_{1}, \dots, c'_{n},f)$ as 
$$ \forall c_1 \dots c_n  \{[\bigwedge_{i=1}^n(c'_i \Rightarrow c_{i})]
 \Rightarrow formula\_prov(f) \}$$
Whenever the truth values of $c_i$ agree with those of the variables of some $S_j$ in $prov(f)$, or 
correspond to a superset of some $S_j$ in $prov(f)$, 
the formula $formula\_prov(f)$ will hold good. 


\item Let $F=\{f_1, f_2, \dots, f_l\}$ be a specification. For each $f_i$, let $prov(f_i)=\{S_{i1}, \dots, S_{ik}\}$ be defined.
Consider a tuple of component parameters $(c'_1, \dots, c'_n)$ and a tuple
of feature parameters $(f'_1, \dots, f'_m)$. Here again, each $c'_i, f'_j$ is a zero or a 1. 
Define $f \_covers(c'_{1}, \dots, c'_{n},f'_{1}, \dots, f'_{m})$ as
$$\bigwedge_{i=1}^m(f'_i \Rightarrow f \_ implements(c'_{1}, \dots, c'_{n},f_i))$$

Define $f \_realizes(c'_{1}, \dots, c'_{n},f'_{1}, \dots, f'_{m})$ as
$$\bigwedge_{i=1}^m(f'_i \Leftrightarrow f \_ implements(c'_{1}, \dots, c'_{n},f_i))$$


\item Let $\Psi=(\overline{\cal F}, \overline{\cal C}, {\cal T})$ be an SPL.
Let $\overline{\cal C} =\{S_1, \dots, S_k\}$. Given a tuple of component parameters $c'_1, \dots, c'_n$ where each $c'_i$ is 0 or 1, 
the predicate $C_I(c'_{1}, \dots, c'_{n})$ is defined as 
$$\bigvee_{j}\bigwedge \limits_{c_i \in Prop(S_j)} c'_i$$
 Then $C_I(c'_{1}, \dots, c'_{n})$ is satisfied iff $\{c'_k \mid c'_k=1\} = S_l$ for some $S_l\in \overline{\cal C}$.
$C_F(f'_1, \dots, f'_m)$ is defined similarly.

\end{enumerate}

\begin{lemma}(Internal Consistency of Traceability)
\label{tcf}
Consider a canonical SPL. 
Let TCF, the trace consistency formula be defined as $\forall c_1 \dots c_n. \bigwedge_{f \in F}[f\_prov(f) \Rightarrow f \_ req(f)]$. Then, ${\cal T}$ is  internally consistent iff TCF is true.
 \end{lemma}

\begin{lemma}(Implements)
\label{impl}
 Given a canonical SPL, a set of components $C$, and a feature $f$, $implements(C, f)$ iff $f \_ implements(c'_{1}, \dots, c'_{n}, f)$ where $Prop(C)=(c'_{1}, \dots, c'_n)$.
\end{lemma}

\begin{lemma}(Realizes, Covers)
\label{rc}
 Given a set of components $C$ and a set of features $F$, let $Prop(C)=(c'_{1}, \dots, c'_{n})$ and $Prop(F)=(f'_{1}, \dots, f'_{m})$. Then the following statements hold:
\begin{enumerate}
 \item $C$ covers $F$ iff $f \_ covers(c'_{1}, \dots, c'_{n}, f'_{1}, \dots, f'_{m})$
\item $C$ realizes $F$ iff $f \_ realizes(c'_{1}, \dots, c'_{n}, f'_{1}, \dots, f'_{m})$
\end{enumerate}
\end{lemma}

\begin{lemma}(Completeness, Soundness)
Given an SPL, the SPL is complete iff \\
$\forall f'_1 \dots f'_m[C_F(f'_1, \dots, f'_m) \Rightarrow \exists c'_1 \dots c'_n[C_I(c'_1, \dots, c'_n) \wedge f\_ covers(c'_1, \dots, c'_n, f'_1, \dots, f'_m)]$

Given an SPL, the SPL is sound iff \\
$\forall c_1 \dots c_n[C_I(c_1, \dots, c_k)] \Rightarrow \exists f_1 \dots f_j[C_F(f_1, \dots, f_j) \wedge f\_ covers(c_1, \dots, c_k, f_1, \dots, f_j)]$
\end{lemma}

\begin{lemma}(Existentially Explicit Features)
 Given a set of features $F$, let $Prop(F)=(f'_1, \dots, f'_m)$. Then $F$ is existentially explicit iff 
$\exists c'_1 \dots c'_n[C_I(c'_1, \dots, c'_n) \wedge f\_realizes(c'_1, \dots, c'_n, f'_1, \dots, f'_m)]$. 
\end{lemma}

\begin{lemma}(Universally Explicit Features)
 Given a set of features $F$, let $Prop(F)=(f'_1, \dots, f'_m)$. Then $F$ is universally explicit iff
$\exists c'_1 \dots c'_n[C_I(c'_1, \dots, c'_n) \wedge f\_realizes(c'_1, \dots, c'_n, f'_1, \dots, f'_m)] \wedge 
\forall c'_1 \dots c'_n\{[(C_I(c'_1, \dots, c'_n) \wedge f\_covers(c'_1, \dots, c'_n, f'_1, \dots, f'_m)]
 \Rightarrow f\_realizes(c'_1, \dots, c'_n, f'_1, \dots, f'_m)\}$.
\end{lemma}

\begin{lemma}(Unique Implementation)
 Given a set of features $F$, let $Prop(F)=(f'_1, \dots, f'_m)$. Then $F$ has a unique implementation iff 
$\exists c'_1 \dots c'_n[C_I(c'_{1}, \dots, c'_{n}) \wedge f\_covers(c'_{1}, \dots, c'_{n}, f'_1, \dots, f'_m)]
\wedge \forall d'_1 \dots d'_n\{[C_I(d'_{1}, \dots, d'_{n}) \wedge$ \\ $  f\_covers(d'_{1}, \dots, d'_n, f'_1, \dots, f'_m)] \Rightarrow (\wedge_{l=1}^n (d'_i \Leftrightarrow c'_i)\}$
\end{lemma}

\begin{lemma}(Common, live and dead elements)
 \begin{enumerate}
  \item A component $c$ is common iff \\
$\forall c'_1, \dots, c'_n, f'_1, \dots, f'_m\{[C_I(c'_{1}, \dots, c'_{n}) \wedge C_F(f'_{1}, \dots, f'_{m})
\wedge f\_covers(c'_{1}, \dots, c'_{n},f'_{1}, \dots, f'_{m})] \Rightarrow c\}$ holds.
\item A component $c$ is live iff \\
$\exists c'_{1}, \dots, c'_{n}, f'_{1}, \dots, f'_{m}\{[C_I(c'_{1}, \dots, c'_{n}) \wedge C_F(f'_{1}, \dots, f'_{m})\wedge f\_covers(c'_{1}, \dots, c'_{n},f'_{1}, \dots, f'_{m}) \wedge c\}$
\item A component $c$ is dead iff \\
$\forall c'_1, \dots, c'_n, f'_1, \dots, f'_m\{[C_I(c'_{1}, \dots, c'_{n}) \wedge C_F(f'_{1}, \dots, f'_{m})
\wedge f\_covers(c'_{1}, \dots, c'_{n},f'_{1}, \dots, f'_{m})] \Rightarrow \neg c\}$ holds.
 \end{enumerate}
\end{lemma}

\begin{lemma}(Superflous)
A component $c_i$ is superflous iff $\forall c'_1, \dots, c'_n,f'_1, \dots,f'_m\{[c'_i \wedge C_I(c'_1, \dots, c'_n) \wedge C_F(f'_1, \dots, f'_m) \wedge f\_covers(c'_1, \dots,c'_i, \dots, c'_n, f'_1, \dots, f'_m)] \Rightarrow \exists d'_1, \dots, d'_n[\neg d'_i \wedge C_I(d'_1, \dots, d'_n) \wedge $ \\ $f\_covers(d'_1, \dots,d'_n, f'_1, \dots, f'_m)]\}$.
\end{lemma}

\begin{lemma}(Redundant)
A component $c_i$ is redundant iff $\forall c'_1, \dots, c'_n f'_1 \dots, f'_m\{[c'_i \wedge C_I(c'_1, \dots, c'_n) \wedge C_F(f'_1, \dots, f'_m) \wedge f \_covers(c'_1, \dots, c'_n, f'_1, \dots, f'_m)] \Rightarrow \exists d'_1 \dots d'_n[\neg d'_i \wedge (\bigwedge_{i=1}^n c'_i \Rightarrow \bigwedge d'_i) \wedge C_I(d'_1, \dots, d'_n) \wedge $ \\ $
f \_covers(d'_1, \dots, d'_n, f'_1, \dots, f'_m)]\}$.
\end{lemma}

\begin{lemma}(Critical)
A component $c$ is critical for $f_j$ iff $\forall c'_1, \dots, c'_n\{[C_I(c'_1, \dots,  c'_n) \wedge f\_implements(c'_1, \dots, c'_n, f_j)] \Rightarrow c\}$. 	
 \end{lemma}

\begin{lemma}(Extends)
Let $F$ and $F'$ be  subsets of features. Let $Prop(F)=(f_1, \dots, f_m)$ and $Prop(F')=(f'_1, \dots, f'_m)$. Then $F'$ extends $F$ iff $\bigwedge_{i=1}^m (f_i \Rightarrow f'_{i})$ is true. $F'$ is extendable iff $\exists f'_1, \dots, f'_m[\bigwedge_{i=1}^m f_i \Rightarrow f'_{i})]$.
\end{lemma}

\begin{thm}
\label{thm:main}
 Given an SPL $\spl$, each of the properties listed in Table \ref{tab_form} holds good iff 
the corresponding formulae evaluate to true.
\end{thm}
\begin{proof}
 The detailed proof is given in the full version of the paper.
\end{proof}

\begin{table*}[hbpt]
	\begin{tabular}{|c|c|}
	\hline
	Properties & Formula \\
	\hline
$Implements(C, f)$ & $f \_ implements(c'_{1}, \dots, c'_{n}, f)$ \\
$Prop(C)= (c'_{1}, \dots, c'_n)$ & \\
\hline
$C$ covers $F$, $Prop(C)=(c'_{1}, \dots, c'_{n})$ & $f \_ covers(c'_{1}, \dots, c'_{n}, f'_{1}, \dots, f'_{m})$\\
$C$ realizes $F$, $Prop(F)=(f'_{1}, \dots, f'_{m})$ & $f \_ realizes(c'_{1}, \dots, c'_{n}, f'_{1}, \dots, f'_{m})$\\
\hline
$\Psi$ complete & $\forall f'_1 \dots f'_m\{C_F(f'_1, \dots, f'_m) \Rightarrow \exists c'_1 \dots c'_n[C_I(c'_1, \dots, c'_n) \wedge  f\_ covers(c'_1, \dots, c'_n, f'_1, \dots, f'_m)]\}$\\ 
\hline
$\Psi$ sound &  $\forall c'_1 \dots c'_n\{C_I(c'_1, \dots, c'_n)] \Rightarrow \exists f'_1 \dots f'_m[C_F(f'_1, \dots, f'_m) \wedge f\_ covers(c'_1, \dots, c'_k, f_1, \dots, f_j)]\}$\\
\hline
$F$ existentially explicit & $\exists c'_1 \dots c'_n[C_I(c'_1, \dots, c'_n) \wedge f\_realizes(c'_1, \dots, c'_n, f'_1, \dots, f'_m)]$ \\
   $Prop(F)= (f'_1, \dots, f'_m)$ & \\
\hline
$F$ universally explicit & $\exists c'_1 \dots c'_n[C_I(c'_1, \dots, c'_n) \wedge
f\_realizes(c'_1, \dots, c'_n, f'_1, \dots, f'_m)] \wedge \forall c'_1 \dots c'_n\{[(C_I(c'_1, \dots, c'_n) \wedge$\\
$Prop(F)=(f'_1, \dots, f'_m)$ & $f\_covers(c'_1, \dots, c'_n, f'_1, \dots, f'_m)]  \Rightarrow  f\_realizes(c'_1, \dots, c'_n, f'_1, \dots, f'_m)\}$.\\
\hline
$F$ has unique implementation & $\exists c'_1 \dots c'_n [C_I(c'_{1}, \dots, c'_{n}) \wedge
f\_covers(c'_{1}, \dots, c'_{n}, f'_1, \dots, f'_m)] \wedge$\\
$Prop(F)=(f'_1, \dots, f'_m)$ &  $\forall d'_1 \dots d'_n\{[C_I(d'_{1}, \dots, d'_{n}) \wedge
  f\_covers(d'_{1}, \dots, d'_n, f'_1, \dots, f'_m)] \Rightarrow (\wedge_{l=1}^n (d'_i \Leftrightarrow c'_i)\}$\\
\hline
$c$ common & $\forall c'_1, \dots, c'_n, f'_1, \dots, f'_m\{[C_I(c'_{1}, \dots, c'_{n}) \wedge
C_F(f'_{1}, \dots, f'_{m}) \wedge
f\_covers(c'_{1}, \dots, c'_{n},f'_{1}, \dots, f'_{m})] \Rightarrow c\}$\\
\hline
$c$ live & $\exists c'_{1}, \dots, c'_{n}, f'_{1}, \dots, f'_{m}\{[C_I(c'_{1}, \dots, c'_{n}) \wedge
 C_F(f'_{1}, \dots, f'_{m})\wedge f\_covers(c'_{1}, \dots, c'_{n},f'_{1}, \dots, f'_{m}) \wedge c\}$\\
\hline
$c$ dead & $\forall c'_1, \dots, c'_n, f'_1, \dots, f'_m\{[C_I(c'_{1}, \dots, c'_{n}) \wedge
C_F(f'_{1}, \dots, f'_{m}) \wedge
f\_covers(c'_{1}, \dots, c'_{n},f'_{1}, \dots, f'_{m})] \Rightarrow \neg c\}$\\
\hline
$c_i$ superfluous & $\forall c'_1, \dots, c'_n,f'_1, \dots,f'_m\{[c'_i \wedge C_I(c'_1, \dots, c'_n)
\wedge C_F(f'_1, \dots, f'_m) \wedge
f\_covers(c'_1, \dots,c'_i, \dots, c'_n, f'_1, \dots, f'_m)] \Rightarrow $\\
 & $\exists d'_1, \dots, d'_n[\neg d'_i \wedge C_I(d'_1, \dots, d'_n) \wedge
 f\_covers(d'_1, \dots,d'_n, f'_1, \dots, f'_m)]\}$\\
\hline
$c_i$ redundant & $\forall c'_1, \dots, c'_n f'_1 \dots, f'_m\{[c'_i \wedge C_I(c'_1, \dots, c'_n) \wedge C_F(f'_1, \dots, f'_m) \wedge f \_covers(c'_1, \dots, c'_n, f'_1, \dots, f'_m)] \Rightarrow$\\
& $\exists d'_1 \dots d'_n[\neg d'_i \wedge (\bigwedge_{i=1}^n c'_i \Rightarrow \bigwedge d'_i) \wedge C_I(d'_1, \dots, d'_n) \wedge f \_covers(d'_1, \dots, d'_n, f'_1, \dots, f'_m)]\}$\\
\hline
$c$ critical for $f_j$ & $\forall c'_1, \dots, c'_n\{[C_I(c'_1, \dots,  c'_n) \wedge f\_implements(c'_1, \dots, c'_n, f_j)] \Rightarrow c\}$\\
\hline
\end{tabular}
\caption{Properties and Formulae}
\label{tab_form}
\end{table*}

\section{Implementation}
In this section, we give some details of the implementation of the theory developed, using off-the-shelf QSAT solvers. We also illustrate the encoding of the analysis problems in QBF and their solutions through a small example.

\subsection{QBF and QDIMACS format}
Quantified Boolean Formulae (QBF) are generalized form of propositional formulae with quantification (existential and universal) over the propositional symbols.
The boolean satisfiability problem for propositional formulae is then naturally extended to QBF satisfiability problem (QSAT). 

Most QBF solvers follow QDimacs, a standard input and output file format. QDimacs Format is built on top of the DIMACS standard for SAT Solver. A QDimacs file representing a QBF has three parts: Preamble, Prefix and Matrix. The notations use a unique indexing of all the propositional variables occurring in the QBF.
\begin{enumerate}

\item $Preamble:$ The Preamble contains different types of information about the file, namely,
	\begin{enumerate}
	\item $Comments:$ Each comment line should start with lower case character 'c'. There can be multiple comment lines in the File.\\
	Format:\\
	{\tt c COMMENT\_STRING}\\
	Example:\\
	{\tt c Testing QBF formulae.}\\
	{\tt c qdimac file for completeness.}

	\item $Problem$ $Line:$ There is only one problem line in each QDimacs File. The problem line starts with the lower case character 'p' followed by the string 'cnf', which denotes that the given formula is in conjunctive normal form (CNF). The 'cnf' string is followed by variables count and clauses count.

	Format:\\
	{\tt p cnf VAR\_COUNT CLA\_COUNT}

	Example:\\
	{\tt p cnf 4 2}
	\end{enumerate}
\item $Prefix:$ The Prefix lines are used to represent the quantifiers in the Formula. Each Prefix line starts with a lower case character 'a' or 'e'; 'a' represents universal quantifier and 'e' represents existential quantifier. Quantifiers are followed by the indices of variables. Each prefix line ends with '0'.

	Example:\\
	{\tt a 1 2 0}\\
	{\tt e 3 4 0}

\item $Matrix:$ Each line in matrix represents a clause and should end with '0'. Each propositional variable in clause is represented by it's corresponding unique index. The complement of a variable is represented by negation of the index.

	Example:\\
	1 3 0\\
	2 -4 0
\end{enumerate}

As an example, the QDimacs format for the formula $\forall X \exists Y ((X \vee \neg Y) \wedge (\neg X  \vee Y)$ is as follows. The first line is a comment line. The second one is the problem line which mentions that there are two variables and two clauses. The third line represents the universal quantification of $X$ and the fourth line represents the existential quantification of $Y$. The fifth line represents the first clause $(X \vee \neg Y)$ and the sixth line represents the second clause $(\neg X  \vee Y)$.
{\tt
\begin{tabbing}
c Illustration \\
p cnf 2 2 \\
a 1 0 \\
e 2 0 \\
1 -2 0 \\
-1 2 0 \\
\end{tabbing}
}
QuBE is a solver for Quantified Boolean Formulas (QBFs). It accepts QBFs in QDimacs format and returns TRUE if the formula is satisfiable, and FALSE otherwise. We have developed a tool called CNF2QDIMAC converter. The tool converts QBFs in CNF to QDimacs format which can be given as input to QuBE. Conversion of arbitrary QBFs to CNF is done using some online tools.

\subsection{An Illustrative Example}
Consider the following SPL $\Psi=(\overline{\cal C}, \overline{\cal F}, {\cal T})$ with 
$\overline{\cal C}=\{\{c_1, c_2\}, \{c_3, c_4\}\}$ and 
$\overline{\cal F}=\{\{f_1,f_2\},\{f_3\}\}$. Thus, there are 4 components and 3 features. Further, let the traceability relation ${\cal T}$ be given as follows: 
\begin{itemize}
\item $prov(f_1)=\{\{c_1, c_2\}, \{c_3\}\}, req(f_1)=\{\{\{c_1\},\{c_3\}\}$
\item $prov(f_2)=\{\{c_2\}\}, req(f_1)=\{\{c_2\}\}$
\item $prov(f_3)=\{\{c_1, c_4\}\}, req(f_3)=\{\{c_4\}\}$
\end{itemize}

Let us answer the following questions using the logic formulation with the help of the QuBE tool.
\begin{enumerate}
\item Does $C=\{c_1, c_2\}$ implement $f_1$? Clearly, the answer is YES. 
In the logic formalism, 
$f\_implements(1,1,0,0,f_1)$ is defined as 
$\forall c_1 c_2 c_3 c_4 \{ [(1 \Rightarrow c_1) \wedge (1 \Rightarrow c_2) \wedge (0 \Rightarrow c_3) \wedge (0 \Rightarrow c_4)] \Rightarrow f\_prov(f_1)\}$ where $f\_prov(f_1)$ = $(c_1 \wedge c_2) \vee c_3$. The formula when simplified is $\forall c_1 c_2 c_3 c_4 ((c_1 \wedge c_2) \Rightarrow ((c_1 \wedge c_2) \vee c_3)$. It is easy to see that the formula evaluates to true. Hence QuBE returns an affirmative answer.

Now consider $C=\{c_3\}$. Does $C$ implement $f_3$? Clearly, the answer is NO. 
In the logic formalism, \\
$f\_implements(0,0,1,0,f_3)$ is defined as 
$\forall c_1 c_2 c_3 c_4 \{ [(0 \Rightarrow c_1) \wedge (0 \Rightarrow c_2) \wedge (1 \Rightarrow c_3) \wedge (0 \Rightarrow c_4)]
\Rightarrow f\_prov(f_3)\}$ where $f\_prov(f_3)$ = 
$(c_1 \wedge c_4)$. The simplified formula is $\forall c_1 c_2 c_3 c_4(c_3 \Rightarrow (c_1 \wedge c_4))$. The assignment $c_3 = 1$, $c_1 = 0$ evaluates the quantifier-free formula to false. Hence QuBE returns a negative answer.

\item Consider $C=\{c_1,c_2\}$. Does C realize $\{f_1,f_2\}$? Clearly, the answer is YES. In the logic formalism, 
$f\_realizes(1,1,0,0,1,1,0)$ is defined as 
$([1 \Leftrightarrow f\_implements(1,1,0,0,f_1)] \wedge [1 \Leftrightarrow f\_implements(1,1,0,0,f_2)] \wedge [0 \Leftrightarrow f\_implements(1,1,0,0,f_3)]
\wedge [0 \Leftrightarrow f\_implements(1,1,0,0,f_4)]$. \\

Now, $f\_implements(1,1,0,0,f_1)$ is defined as 
$\forall c_1 c_2 c_3 c_4  \{[(1 \Rightarrow c_1) \wedge (1 \Rightarrow c_2) \wedge (0 \Rightarrow c_3) \wedge (0 \Rightarrow c_4)] \Rightarrow f\_prov(f_1)\}$ where $f\_prov(f_1)$ is defined as 
$(c_1 \wedge c_2) \vee c_3$. Clearly, $f\_implements(1,1,0,0,f_1)$ holds. Thus, we have 
$[1 \Leftrightarrow f\_implements(1,1,0,0,f_1)]$ is true. Similarly, it can be seen that 
$[1 \Leftrightarrow f\_implements(1,1,0,0,f_2)]$ is true.\\

Likewise, $f\_implements(1,1,0,0,f_3)$ is 
$\forall c_1 c_2 c_3 c_4  [(1\Rightarrow c_1) \wedge (1 \Rightarrow c_2) \wedge (0 \Rightarrow c_3) \wedge (0 \Rightarrow c_4)] \Rightarrow (c_1 \wedge c_4)$, which is false.  Hence, 
$[0 \Leftrightarrow f\_implements(1,1,0,0,f_3)]$ is true. \\

Similarly,
 $f\_implements(1,1,0,0,f_4)$ is
$\forall c_1 c_2 c_3 c_4 \{ [(1 \Rightarrow c_1) \wedge (1 \Rightarrow c_2) \wedge (0 \Rightarrow c_3) \wedge (0 \Rightarrow c_4)]  \Rightarrow (c_4)\}$, which is false. Hence, $[0 \Leftrightarrow f\_implements(1,1,0,0,f_4)]$ is true. 
Thus, we have the answer true from QuBE. 

Now consider the question: does $C$ realize $f_1$? Clearly, $C$ covers $f_1$, but realizes $\{f_1, f_2\}$. Again, the logic formalism for the same is
 $f\_realizes(1,1,0,0,1,0,0)$, which is defined as 
$([1 \Leftrightarrow f\_implements(1,1,0,0,f_1)] \wedge [0 \Leftrightarrow f\_implements(1,1,0,0,f_2)] \wedge [0 \Leftrightarrow f\_implements(1,1,0,0,f_3)]
\wedge [0 \Leftrightarrow f\_implements(1,1,0,0,f_4)]$. 

As seen above, clearly, $[1 \Leftrightarrow f\_implements(1,1,0,0,f_1)]$ holds. However, 
we have  $f\_implements(1,1,0,0,f_2)$ is true since $prov(f_2)=\{\{c_2\}\}$. Then, we do not have $[0 \Leftrightarrow f\_implements(1,1,0,0,f_2)]$. Therefore, QuBE returns false. 

\item Is the given SPL complete? That is, for every $F \in \overline{\cal F}$, does there exist some $C \in \overline{\cal C}$ such that $Covers(C,F)$? 
Clearly, the answer is NO since there is no $C \in \overline{\cal C}$ covering $\{f_3\} \in \overline{\cal F}$. The formula for this is 
$\forall f'_1 f'_2 f'_3[C_F(f'_1, f'_2, f'_3) \Rightarrow \exists c'_1 c'_2 c'_3 c'_4[C_I(c'_1, \dots, c'_4) \wedge f\_ covers(c'_1, \dots, c'_4, f'_1, \dots, f'_3)]$. This expands out to 

$C_F(1,1,1) \Rightarrow \exists c'_1, c'_2 c'_3 c'_4[C_I(c'_1, \dots, c'_4) \wedge f\_ covers(c'_1, \dots, c'_4, 1,1,1)]$ and 

$C_F(1,1,0) \Rightarrow \exists c'_1, c'_2 c'_3 c'_4[C_I(c'_1, \dots, c'_4) \wedge f\_ covers(c'_1, \dots, c'_4, 1,1,0)]$ and 

$C_F(1,0,1) \Rightarrow \exists c'_1, c'_2 c'_3 c'_4[C_I(c'_1, \dots, c'_4) \wedge f\_ covers(c'_1, \dots, c'_4, 1,0,1)]$ and 

$C_F(0,1,1) \Rightarrow \exists c'_1, c'_2 c'_3 c'_4[C_I(c'_1, \dots, c'_4) \wedge f\_ covers(c'_1, \dots, c'_4, 0,1,1)]$ and 

$C_F(0,0,1) \Rightarrow \exists c'_1, c'_2 c'_3 c'_4[C_I(c'_1, \dots, c'_4) \wedge f\_ covers(c'_1, \dots, c'_4, 0,0,1)]$ and 

$C_F(0,1,0) \Rightarrow \exists c'_1, c'_2 c'_3 c'_4[C_I(c'_1, \dots, c'_4) \wedge f\_ covers(c'_1, \dots, c'_4, 0,1,0)]$ and

$C_F(1,0,0) \Rightarrow \exists c'_1, c'_2 c'_3 c'_4[C_I(c'_1, \dots, c'_4) \wedge f\_ covers(c'_1, \dots, c'_4, 1,0,0)]$ and 

$C_F(0,0,0) \Rightarrow \exists c'_1, c'_2 c'_3 c'_4[C_I(c'_1, \dots, c'_4) \wedge f\_ covers(c'_1, \dots, c'_4, 0,0,0)]$. 

Among these, $C_F(1,1,0)$, $C_F(0,0,1)$ evaluates to true. The rest evaluate to false - hence the formula involving them holds.

Now, consider $C_F(1,1,0)$. Then we must check 
whether $\exists c'_1 c'_2 c'_3 c'_4[C_I(c'_1, \dots, c'_4) \wedge f\_ covers(c'_1, \dots, c'_4, 1,1,0)]$ holds. The tuple $(1,1,0,0)$ as well as $(0,0,1,1)$ satisfy
$C_I(c'_1,c'_2,c'_3,c'_4)$. Hence, these are the only two tuples that we need to examine 
for $(c'_1,c'_2,c'_3,c'_4)$. Consider $(1,1,0,0)$. Then 
$[C_I(1,1,0,0) \wedge f\_ covers(1,1,0,0, 1,1,0)]$
evaluates to $true \wedge [1 \Rightarrow f\_implements(1,1,0,0,f_1)]
\wedge  [1 \Rightarrow f\_implements(1,1,0,0,f_2)]
\wedge$
\\
$[0 \Rightarrow f\_implements(1,1,0,0,f_3)]$. 
Clearly, this is true, as $\{c_1, c_2\}$ covers $\{f_1, f_2\}$. 

Now consider 
$C_F(0,0,1)$. Then we must check
$\exists c'_1 c'_2 c'_3 c'_4[C_I(c'_1, \dots, c'_4) \wedge f\_ covers(c'_1, \dots, c'_4, 0,0,1)]$ holds. Again, consider the two possibilities for 
$C_I(c'_1,c'_2,c'_3,c'_4)$. Look at 
$C_I(1,1,0,0)$ first. Then we have to check if 
$f\_ covers(1,1,0,0, 0,0,1)$ is true. This is
$[0 \Rightarrow f\_implements(1,1,0,0,f_1)]
\wedge  [0 \Rightarrow f\_implements(1,1,0,0,f_2)]
\wedge [1 \Rightarrow f\_implements(1,1,0,0,f_3)]$. 
Clearly, $f\_implements(1,1,0,0,f_3)$ does not hold since $prov(f_3)=\{c_1, c_4\}$ and 
$c_4$ can be assigned 0 in this formula. 
 Now consider the second assignment $(0,0,1,1)$. Then again, $C_I(0,0,1,1)$ holds. Now check if
$f\_ covers(0,0,1,1, 0,0,1)$ holds. That is, 
$[0 \Rightarrow f\_implements(0,0,1,1,f_1)]
\wedge  [0 \Rightarrow f\_implements(0,0,1,1,f_2)]
\wedge [1 \Rightarrow f\_implements(0,0,1,1,f_3)]$. 
Since $prov(f_3)=\{\{c_1,c_4\}\}$, $f\_implements(0,0,1,1,f_3)$ is false. Thus, this does not hold good as well. 

Therefore, for $\{f_3\}$ (equivalently, $C_F(0,0,1)$), there is no $C_I(c'_1,c'_2,c'_3,c'_4)$ which realizes $\{f_3\}$. Hence, QuBE returns false. Then, we can conclude that the SPL is not complete.
\end{enumerate}

\section{Results of Analyses on The ECPL Case-study}
\label{sec_results}

In this section, we analyze some properties of the ECPL example using QUBE. The platform $\overline{\cal C}$ contains 
the following architectures: 
\begin{enumerate}
 \item $C_1=\{$ Door Lock Manager, Unlock Driver Door, Unlock all doors, Lock all doors$\}$
\item $C_2=\{$Door lock manager, Unlock driver door, Unlock all doors, Lock all
doors, AutoLock, Speed$\}$
\item $C_3=\{$ Door lock manager, Unlock driver door, Unlock all doors, Lock all
doors, AutoLock, Gear in park $\}$
\item $C_4=\{$ Door lock manager, Unlock driver door, Unlock all doors, Lock all
doors, Power Lock, Courtesy switch, Key signal, silldoor signal,
Automatic$\}$
\item $C_5=\{$ Door lock manager, Unlock driver door, Unlock all doors, Lock all
doors, Power Lock, Courtesy switch, Key signal, silldoor signal,
Manual$\}$
\item $C_6=\{$ Door lock manager, Unlock driver door, Unlock all doors, Lock all
doors, AutoLock, Speed, Power Lock, Courtesy switch, Key signal,
silldoor signal, Automatic$\}$
\item $C_7=\{$ Door lock manager, Unlock driver door, Unlock all doors, Lock all
doors, AutoLock, Speed, Power Lock, Courtesy switch, Key signal,
silldoor signal, Manual$\}$
\item $C_8=\{$Door lock manager, Unlock driver door, Unlock all doors, Lock all
doors, AutoLock, Gear in park, Power Lock, Courtesy switch, Key
signal, silldoor signal, Automatic$\}$
\item $C_9=\{$ Door lock manager, Unlock driver door, Unlock all doors, Lock all
doors, AutoLock, Gear in park, Power Lock, Courtesy switch, Key
signal, silldoor signal, Manual$\}$
\end{enumerate}

Consider the following specifications in the scope $\overline{\cal F}$.
\begin{enumerate}
 \item $F_1=\{$Power Lock, f\_automatic$\}$ 
\item $F_2=\{$Power Lock, f\_automatic, Door Lock, Shift outof Park, Door relock$\}$.
\end{enumerate}

\begin{enumerate}
 \item Does $C_1$ realize $F_1$?
The formula to check is 
[1 $\Leftrightarrow f\_implements(1,1,1,1,0,\dots,0$,PowerLock)] $\wedge$ [1 $\Leftrightarrow f\_implements(1,1,1,1,0,\dots,0$, f\_automatic)] $\wedge \dots$ 
[0 $\Leftrightarrow f\_implements(1,1,1,1,0,\dots,0$,Door relock)] \\

Lets look at $f\_implements(1,1,1,1,0,\dots,0$,PowerLock). 
Let $c_1=Door Lock Manager$, $c_2=UnLock Driver Door$, $c_3=Unlock all doors$, $c_4=Lock all doors$, 
$c_5=PowerLock$. 
This is defined as 
$\forall c_1, \dots, c_n\{([1 \Rightarrow c_1] \wedge [1 \Rightarrow c_2] \wedge [1 \Rightarrow c_3]
\wedge [1 \Rightarrow c_4] \wedge [0 \Rightarrow c_5] \dots [0 \Rightarrow c_n]) \Rightarrow (c_1 \wedge c_5)\}$. 
Clearly, this does not hold (for $c_5=0$, the formula does not hold).

Hence, QUBE returns false. 

\item Is ECPL sound? If so, then for every $C_i \in \overline{\cal C}$, we can find a specification $F_i$ such that 
$Covers(C_i, F_i)$. The formulae for this is
\begin{tabbing} 
XXXXX\=\kill
$\forall c_1 \dots c_n[C_I(c_1, \dots, c_n)] \Rightarrow$ \\
\> $\exists f_1 \dots f_m[C_F(f_1, \dots, f_m) \wedge$\\
\> $f\_ covers(c_1, \dots, c_n, f_1, \dots, f_m)]$
\end{tabbing}

Consider the tuple $(1,1,1,1,0, \dots, 0)$ where the first four entries are 1, and the rest are zero.
This corresponds to $C_1$. Clearly, $C_I(1,1,1,1,0, \dots,0)$. Lets look at 
$f\_ covers(1,1,1,1,0, \dots,0, f'_1, \dots, f'_m)$. It is easy to see that 
$f\_implements(1,1,1,1,0,\dots,0,f)$ does not hold good for any $f$ since 
$c_1=Door Lock Manager$ does not provide any features alone, and $c_i, i >0$ 
do not provide any features. Thus,  the formula does not hold good, and QUBE returns false. 
Hence, the ECPL is not sound.

\item Is $F_1$ universally explicit? If so, then any $C_i \in \overline{\cal C}$ which covers $F_1$ must realize $F_1$; moreover, there must be atleast one $C \in \overline{\cal C}$ which covers it.
The formula for this is

$\exists c'_1 \dots c'_n[C_I(c'_1, \dots, c'_n) \wedge f\_realizes(c'_1, \dots, c'_n, f'_1, \dots, f'_m)] \wedge 
\forall c'_1 \dots c'_n\{[(C_I(c'_1, \dots, c'_n) \wedge f\_covers(c'_1, \dots, c'_n, f'_1, \dots, f'_m)]
 \Rightarrow f\_realizes(c'_1, \dots, c'_n, f'_1, \dots, f'_m)\}$. 

 Let us denote 
$c_1$=Door lock manager, $c_2$=AutoLock, $c_3$=Power Lock, $c_4$=Gear in Park and $c_5$=Automatic, $c_6$=Unlock driver door, $c_7$=Unlock all doors, $c_8$=Lock all doors, $c_9$=Courtsey switch, $c_{10}$=Key signal, 
$c_{11}$=sill door signal, $c_{12}$=Speed and $c_{13}$=Manual. 
Similarly, let $f_1$=Power Lock and $f_2$=f\_automatic. 
Consider the component tuple $(1,0,1,0,1,1,1,1,1,1,1,0,0)$. \\

Then we have 
$C_I(1,0,1,0,1,1,1,1,1,1,1,0,0)$.($C_4$ corresponds to this set) and 
$f\_realizes(1,0,1,0,1,1,1,1,1,1,1,0,0,1,1,0, \dots,0)$ ($C_4$ realizes $F_1$). Corresponding to this tuple, 
Consider the component tuple $(1,1,1,1,1,1,1,1,1,1,1,0,0)$. Clearly, 
$C_I((1,1,1,1,1,1,1,1,1,1,1,0,0)$ ($C_8$ corresponds to this). As $C_4 \subseteq C_8$, $C_8$ 
covers $F_1$. However, $f\_realizes(1,1,1,1,1,1,1,1,1,1,1,0,0,1,1,0, \dots,0)$ does not hold  
since :\\

Consider $0 \Leftrightarrow$\\
 $f\_implements((1,1,1,1,1,1,1,$
 $1,1,1,1,0,0, Shiftout of Park),$\\
 a conjunct in\\
 $f\_realizes(1,1,1,1,1,1,1,1,1,1,1,0,0,1,1,0, \dots,0)$. \\Now, 
it can be seen that $f\_implements((1,1,1,1,1,1,1,1,1,1,1,0,0$, Shiftout of Park) holds, since 
the component Gear in Park provides the feature Shift out of Park. Hence, 
this conjunct does not hold good. Hence,  
$f\_realizes(1,1,1,1,1,1,1,1,1,1,1,0,0,1,1,0, \dots,0)$ does not hold. \\

Hence, QUBE returns false. 
Thus, for the component tuple $(1,0,1,0,1,1,1,1,1,1,1,0,0)$ which realizes $F_1$,
 there exists a component tuple which covers, but does not realize $F_1$. Hence, $F_1$ is not universally explicit.

\end{enumerate}

\section{Conclusion}
\label{sec_con}
In this report, we have given a new definition for products in a Software Product Line, based on the notion of derivability of feature specifications from component architectures. The traceability relation between features and components plays a central role in this definition. We show that our definition is different from the consistency based definition of SPL products and captures the implementation relation in a more natural way. In the light of this, we define a set of analysis problems for the SPLs. We show that these problems can be formulated as Quantified Boolean Formulae and can be solved using QSAT tools such as QUBE.

We have demonstrated the feasibility of our approach through a small fragment of an industrial SPL. The scalability of the above approach for complete SPLs is yet to be studied. Since QSAT problem is PSPACE-complete, generic QSAT solvers may not scale well. However, one observes that the formulas for the analyses have very specific structure which can be exploited for efficient QSAT solving.

The proposed semantic model of the SPL treats specifications and architectures as sets of features and components respectively. When richer structure is imposed on these elements, it will affect the definition of traceability relation. Then the implementation relation has to be  refined to handle the resulting complexity.



\bibliographystyle{IEEEtran}
\bibliography{mainASE2010}


\end{document}